\documentclass[onecolumn]{IEEEtran}
\usepackage{amsmath,amsfonts,amssymb,euscript, graphicx, 
epsfig,
enumerate,float,afterpage, subfigure, ifthen}%

\usepackage{url}
\usepackage{hyperref}

\newtheorem{thm}{Theorem}
\newtheorem{cor}{Corollary}
\newtheorem{lem}{Lemma}

\newcommand{\script}[1]{{{\cal{#1} }}}

\newcommand{\busy}{\mbox{busy}}

\begin{document}

\title
 % {A Coffee-Shop Model for Network Layer Multi-Channel Multiple Access}
 {Reversible Models for Wireless Multi-Channel Multiple Access} 
\author{Michael J. Neely \\ University of Southern California
\thanks{This work was supported by grant NSF SpecEES 1824418.}
}

\markboth{Posted to arxiv September 14, 2019}{Neely}

\maketitle

\begin{abstract} 
This paper presents a network layer model for a wireless multiple access system with both persistent and non-persistent users.  There is a single access point with multiple identical channels.    Each user who wants to send a file first scans a subset of the channels to find one that is idle.   
If at least one idle channel is found, the user transmits a file over that channel.   If no idle channel is found,  a persistent user will repeat the access attempt at a later time, while a non-persistent user will  leave.  This is a useful mathematical model for 
``coffee shop'' situations where a group of 
persistent users stay near an access point for an extended period of time while 
non-persistent users come and go.  
 Users have heterogeneous activity behavior, file upload rates, and 
service durations. The system is a complex multi-dimensional Markov chain.  
The steady state probabilities are found by exploiting a latent reversibility 
property.  This enables simple expressions for throughput and blocking 
probability. 
\end{abstract} 

\section{Introduction} 

Consider a wireless system with a single access point that has $m$ identical channels, where $m$ is a positive integer.  
Each channel can support one file transmission, so that up to $m$ files can be transmitted
simultaneously.  Different types of users want to upload files to the access point.  To do this, they first need to find an idle channel.  At the start of every upload attempt, a user randomly scans a subset of the channels in hopes of finding at least one channel that is idle.  Let $s$ be the size of the subset that is scanned and assume that $1\leq s \leq m$. Example numbers are  $m=25$ and $s=5$, so that every user randomly scans 5 of the 25 channels.   If an idle channel is found,
the user sends a single file over that channel (if multiple idle channels are found, the choice of which one to use is made arbitrarily).   If no idle channel is found, the users react differently depending on their type:  Persistent users try again later, while non-persistent users leave and do not return.

 An example of this situation is when the access point is in a fixed location, such as in a coffee shop.   Persistent users are customers who find a table at the coffee shop, stay for an extended period of time, and use their wireless devices during their stay.  
  Non-persistent users either walk past the coffee shop without entering, or enter only for a short time (perhaps to place a take-out order).  The goal of this paper is to establish a Markov chain model for this system and to analyze the model to obtain steady state behavior for throughput and blocking probability.

The system operates in continuous time over the timeline $t \geq 0$.  A channel $j \in \{1, \ldots, m\}$ is said to be \emph{busy} at time $t$ if it is currently being used, that is, if there is a user that is transmitting a file over that channel. Let $B(t) \in \{0,1, \ldots, m\}$ 
be the total number of channels that are busy at time $t$. Suppose a user who is not currently transmitting attempts to access the network at a time $t$ for which $B(t)=b$.  This user scans a random subset of $s$ channels to determine which (if any) are not being used, with all subsets equally likely.  Let $\theta(b)$ denote the conditional
probability that a user successfully finds an idle channel, given that $B(t)=b$.  For example, if $s=1$ then
$\theta(b) = 1 -  b/m$ for all  $b \in \{0, 1, \ldots, m\}$. If $s \in \{1, \ldots, m\}$ then 
\begin{equation} \label{eq:example-success-probability} 
\theta(b) = \left\{ \begin{array}{ll}
1 &\mbox{if $b \in  \{0, 1, \ldots, s-1\}$} \\
1-\left(\frac{b}{m}\right)\left(\frac{b-1}{m-1}\right)\cdots\left(\frac{b-(s-1)}{m-(s-1)}\right)  & \mbox{if $b \in  \{s, \ldots, m\}$} 
\end{array}
\right.
\end{equation} 
An interesting feature of this model is that the probability of successfully finding 
an idle channel at time $t$ depends only on $B(t)$, not on which types of users 
are using each channel.  However, $B(t)$ is not enough to describe the state of the system: The full system state is described by a multi-dimensional Continuous Time Markov Chain (CTMC). 

First, this paper considers the case when all users are non-persistent.  There are $k$ different
classes of non-persistent users. Users for each class arrive according to independent Poisson processes with rates $\lambda_1, \ldots, \lambda_k$.  File sizes are independent and exponentially distributed with rate $\mu_i$ for each class $i \in \{1, \ldots, k\}$.  The system state is  
$X(t) = (X_1(t), \ldots, X_k(t))$,
where $X_i(t)$ is the number of class $i$ users currently transmitting over the network, and 
$B(t) = \sum_{i=1}^k X_i(t)$. The state space grows exponentially with $k$.  Fortunately, the system exhibits a latent \emph{reversible} property. A general theory of reversibility in Markov chains is described in \cite{kelly-reversibility}.   Our CTMC is similar to (but not the same as) the \emph{open migration processes} described in \cite{kelly-reversibility}.   This similarity motivates us to guess a particular product-form steady state distribution. The guess is validated by showing that it satisfies the \emph{detailed balance equations}. This leads to a closed form expression for the steady state mass function in terms of the $\lambda_i$ and $\mu_j$ values.  This yields several insightful results, including an expression for the long term access success probability that depends only on the number of channels $m$, the number of channels scanned $s$, and a 
single \emph{loading parameter} $\rho= \sum_{i=1}^k \lambda_i/\mu_j$.

\begin{figure}[t]
   \centering
   \includegraphics[height=1in]{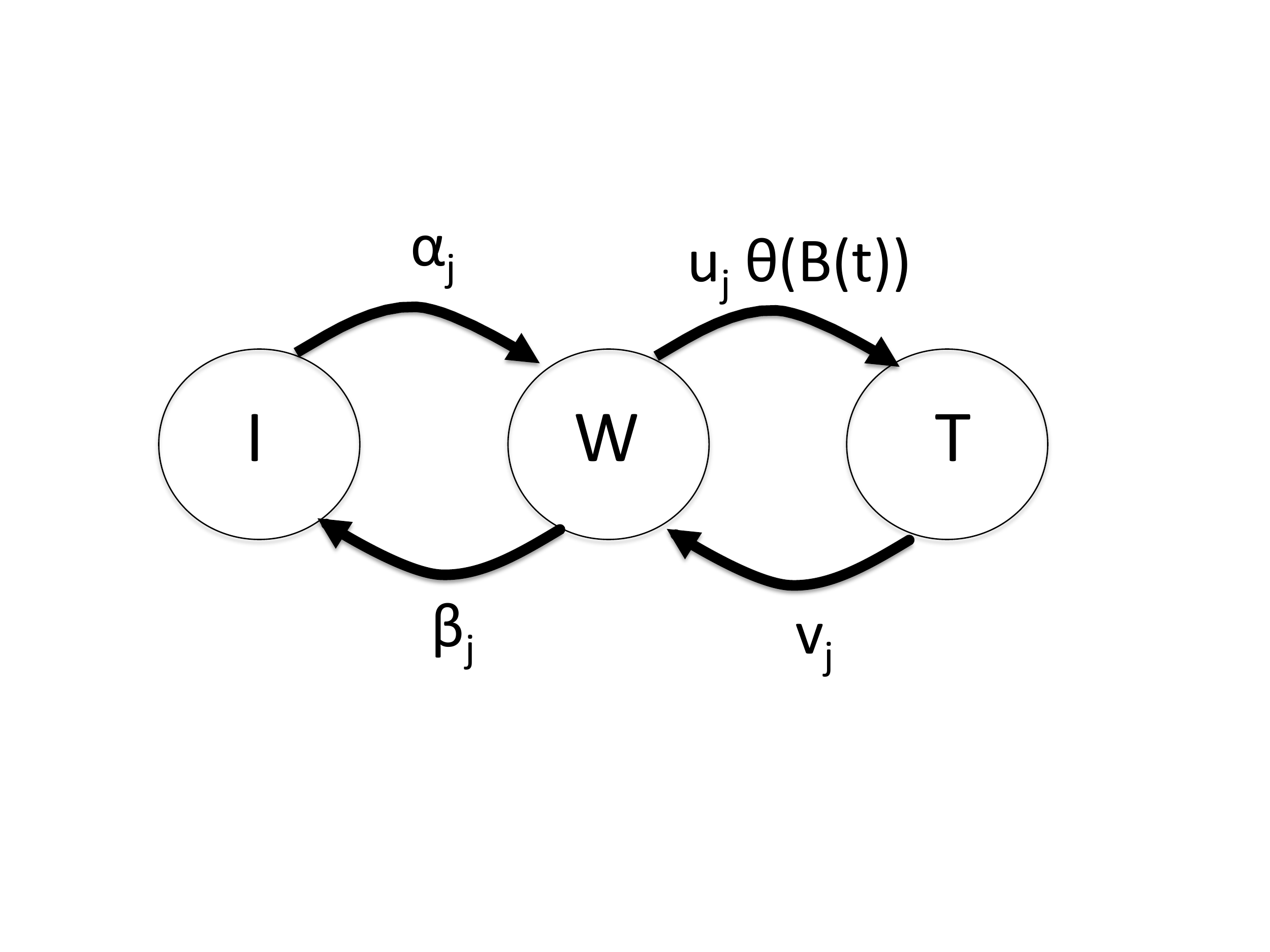} % requires the graphicx package
   \caption{A state diagram showing the idle, waiting, and transmitting states
   for persistent user $j \in \{1, \ldots, n\}$.}
   \label{fig:3-state}
\end{figure}

Next, this paper considers the scenario with both persistent and non-persistent users. 
As before, $(X_1(t), \ldots, X_k(t))$ is the vector that specifies the number of 
non-persistent user of each type that are currently transmitting.  In addition, 
there are $n$ persistent users, where $n$ is a  
positive integer. 
 Each persistent user $j \in \{1,\ldots, n\}$ can be in one of three \emph{activity states}: Idle, Waiting, Transmitting.  Let $A_j(t) \in \{I, W, T\}$ denote the current activity state of persistent user $j \in \{1, \ldots, n\}$. The dynamics of $A_j(t)$ are described by the 
 3-state diagram of Fig. \ref{fig:3-state}: 
 \begin{itemize} 
 \item Idle ($A_j(t)=I$):  Persistent user $j$ is not using its wireless device and thus
 has no files to send. This user stays in the idle state for an independent and exponentially distributed time with parameter $\alpha_j$.  
 \item Waiting ($A_j(t)=W$): Persistent user $j$ is waiting to either attempt transmission of a new file, or go to the idle state, according to independent racing exponential variables of parameters $\beta_j$ and $u_j$.   If the user attempts a transmission but fails, it remains in the waiting state. 
 \item Transmitting ($A_j(t)=T$): Persistent user $j$ is currently transmitting a file over a channel.  This user 
 stays in the transmitting state for an independent and exponentially distributed time with parameter $v_j$.
  \end{itemize} 
  The system state is $W(t)=(X_1(t), \ldots, X_k(t); A_1(t), \ldots, A_n(t))$. The 
  total number of busy channels is 
  $$B(t) = \sum_{i=1}^k X_i(t)  + \sum_{j=1}^n 1_{\{A_j(t)=T\}}$$ 
  where 
 $1_{\{A_j(t)=T\}}$ is an indicator function that is $1$ if $A_j(t)=T$ and $0$ else.  A key aspect of Fig. \ref{fig:3-state} is that the transitions $W \rightarrow T$ for persistent user $j$ 
 occur with transition rate $u_j \theta(B(t))$, which multiplies the access attempt rate $u_j$ by the current success probability $\theta(B(t))$.  In particular, this transition rate 
 depends on the current number of  busy channels, which depends on the history of events associated with \emph{all users}.  Thus, the 3-state ``mini-chain'' of Fig. \ref{fig:3-state} 
 for each persistent user $j \in \{1, \ldots, n\}$ 
 is only a partial view of the larger CTMC: These $n$ different  
 ``mini-chains'' are coupled in a nontrivial way.   Using insight obtained from the 
 case when all users are non-persistent, we guess that the system is reversible, guess a particular steady state structure, and verify these guesses by showing that the detailed balance equations hold.  The resulting steady state probabilities for this case have a simple product form solution.  Unfortunately, 
 there are exponentially many states (the number of states can easily be larger than the current estimate on the number of atoms in the universe) 
 and it is not obvious how to sum the joint probabilities to obtain individual user performance.   We provide a polynomial-time method for computing the exact sums by using a discrete Fourier transform. 
 
 \subsection{Related work} 
 
Our system model is similar to recent work in  \cite{lei-mean-field-mac} that also treats multi-channel systems that scan a subset of the $m$ available 
 channels before transmission.  The work in \cite{lei-mean-field-mac} also assumes that users are in one of three states 
 (idle, probing, transmitting), which is similar to our 3-state persistent user structure.  The work in  \cite{lei-mean-field-mac} does not solve the resulting steady state probabilities, rather, it develops mean-field results that are asymptotically accurate  when all users have identical parameters and when the network size scales to infinity.   In contrast, our work   provides the exact steady state values for the continuous time Markov chain for any system size and for heterogeneous user parameters.  It also treats the case when both persistent and non-persistent users are present.  It should be emphasized that our work exploits a latent reversibility property that does not exist in the model of \cite{lei-mean-field-mac}.  In particular the 3-state user dynamics of \cite{lei-mean-field-mac} can roughly be viewed as similar to those of Fig. \ref{fig:3-state} with the exception that there is no $W \rightarrow I$ transition, and the $T\rightarrow W$ transition is replaced with a $T \rightarrow I$ transition.  Intuitively it is clear that if it is possible to have a transition $I\rightarrow W$ but impossible to have a transition in the opposite direction, then reversibility fails.  It is not clear if exact steady state behavior can be obtained when reversibility fails; that remains an important open question and mean field analysis is an important technique for those situations. 
 
 Our model of persistent users accounts for heterogeneous \emph{human user activity}, 
 where users can be in various states depending on their activity patterns. 
 The topic of mathematical models for human-based activity patterns for wireless communication 
 is of recent interest. 
For example,  related Markov-based models of human user activity and human response times 
are treated in \cite{atilla-user-behavior} for wireless scheduling; related 2-state user activity
models are used in  \cite{xiaohan-file-download-ton} to treat file downloading as a constrained restless 
bandit problem.

\section{Non-persistent users}  \label{section:non-persistent}

This section considers the case where all users are non-persistent.  Each user arrives once and makes one attempt to access a channel.   If the access is successful then the user transmits its file, else it leaves and does not return.  This is a useful model for highly mobile wireless systems that pass by an access point for a short time.  In the coffee shop example, these are users that either walk past the coffee shop but do not enter, or enter the shop and stand in line for a take-out order but do not stay for long. If they cannot obtain access after one attempt, they do not try again.
 
Assume there are $k$ classes of such users, where $k$ is a positive integer.  Users from each class $i \in \{1, ..., k\}$ arrive according to independent Poisson processes with rates $\lambda_1, \ldots, \lambda_k$.   Each user has one file to send.  File service times are independent. Files from class $i$ users have service times that 
are exponentially distributed with parameter $\mu_i$.  Assume that $\lambda_i>0$ and $\mu_i>0$ for all $i \in \{1, \ldots, k\}$. 
The different classes can be used to represent different communities of users who may have different arrival rate and file size parameters.

\subsection{Markov chain model} 

The system can be modeled as a continuous time Markov chain (CTMC) with vector 
state $X(t)=(X_1(t), ..., X_k(t))$, where $X_i(t)$ is the number of type $i$ files currently transmitting at time $t$. The state space $\mathcal{S}$ is given by the set of all vectors $(x_1, ..., x_k)$ that have nonnegative integer components such that $\sum_{i=1}^k x_i \leq m$, where $m$ is the number of channels (assume $m$ is a positive integer).  Let $B(t) = \sum_{i=1}^n X_i(t)$ be the number of busy channels.   A user that arrives to the system scans a subset of the channels to find one that is idle. 
For each $b \in  \{0, 1, ..., m\}$ define $\theta(b)$ as the conditional probability that a newly arriving 
user finds an available channel, given that $B(t)=b$. The value of $\theta(b)$ associated with finding at least one idle channel in a system with $s$ busy channels and $m$ total channels is given in \eqref{eq:example-success-probability}. 
We shall call $\theta(b)$ the \emph{conditional success probability function}.  Our mathematical analysis does not require $\theta(b)$ to have the form \eqref{eq:example-success-probability} and allows for more general success probability functions.  We assume only that $\theta(b)$ satisfies the following basic properties: 
\begin{align}
&0 \leq \theta(b) \leq 1 \quad \forall b \in  \{0, 1, 2, \ldots, m\} \label{eq:sp-1}  \\
&\theta(b) > 0 \quad \forall b \in \{0, 1, 2, \ldots, m-1\} \label{eq:sp-2} \\
&\theta(m)=0 \label{eq:sp-3} 
\end{align}
Requirement \eqref{eq:sp-1} ensures $\theta(b)$ is a valid probability for each $b \in  \{0, 1, 2, \ldots, m\}$; 
requirement \eqref{eq:sp-2} ensures that it is possible to utilize all $m$ channels simultaneously (for example, if this were violated by having  $\theta(5)=0$ but $m=10$, then a system that is initially empty could never have more than 5 active channels, which under-utilizes the existing 10 channels);  requirement \eqref{eq:sp-3} enforces the physical constraint that the system cannot support more than $m$ active channels simultaneously. The particular success probability function in \eqref{eq:example-success-probability} indeed satisfies \eqref{eq:sp-1}-\eqref{eq:sp-3}.

To completely describe the Markov chain structure of this system, it remains to specify the transition rates. 
The transition rates $q_{w,z}$ between two states $w=(x_1, ..., x_k)$ and $z=(y_1, ..., y_k)$ are as follows: Fix an integer $j \in \{1, ..., k\}$ and define $e_j = (0, 0, ..., 0, 1, 0, ..., 0)$ as the vector that is $1$ in entry $j$ and zero in all other entries.  Let $x=(x_1, ..., x_k)$ and $x + e_j = (x_1, ..., x_j+1, .., x_k)$ be two states in the state space $\mathcal{S}$.  Then
\begin{itemize} 
\item Transition rate $x \rightarrow x+e_j$ is given by 
$$q_{x, x+e_j} = \lambda_j\theta\left(\sum_{i=1}^kx_i\right)$$ 
This is the product of the arrival rate $\lambda_j$ with the success probability given that the new 
user scans when the system state is $x=(x_1, ..., x_k)$. 
\item Transition rate $x+e_j \rightarrow x$ is given by 
$$q_{x+e_j, x} = (x_j+1) \mu_j$$ 
This is because there are currently $(x_j+1)$ jobs of type $j$ that are actively using channels, and each has an exponential service rate equal to $\mu_j$. 
\end{itemize} 

Since the system state can change by at most one at any instant of time, there are no other types of transitions and so $q_{w,z}=0$ for states $w, z \in \mathcal{S}$ that do not have the above form. It is not difficult to see that the Markov chain is \emph{irreducible}, so that it is possible to get from any state of the state space $\script{S}$ to any
other state in $\script{S}$  (the requirement \eqref{eq:sp-2} and the fact that $\lambda_i>0$ for all $i \in \{1, \ldots, k\}$  ensure this). 

\subsection{Basic Markov chain theory} 

This subsection recalls basic Markov chain theory (see, 
for example, \cite{kelly-reversibility}\cite{gallager}\cite{bertsekas-data-nets}). 
Consider a continuous time Markov chain (CTMC) with a finite or countably infinite state space 
$\script{S}$ and transition rates $q_{w,z}\geq 0$ for all $w,z \in \script{S}$.   Assume that $q_{w,w}=0$ for all $w \in \script{S}$. The states of $\script{S}$ can be viewed as nodes of a graph; the links of the graph are defined by state-pairs $(w,z)$ such that $q_{w,z}>0$; the CTMC is said to be \emph{irreducible} if this graph has a path from every node to every other node.  
A \emph{probability mass function} over the state space $\script{S}$ is a vector 
$(p(w))_{w \in \script{S}}$ that satisfies $p(w)\geq 0$ for all $w \in \script{S}$ and $\sum_{w \in \script{S}} p(w)=1$. 
The goal is to find a mass function that satisfies the following \emph{global balance equations}: 
\begin{equation} \label{eq:GBE} 
p(w) \sum_{z \in \script{S}} q_{w,z}  = \sum_{z \in \script{S}} p(z) q_{z,w} \quad \forall w \in \script{S} 
\end{equation} 
It is well known that if the CTMC is irreducible, then there is at most one probability mass function
$(p(w))_{w \in \script{S}}$ that solves \eqref{eq:GBE}.  If such a mass function exists, then 
it is the unique steady state mass function for the CTMC. If the CTMC is irreducible and has a finite state space, then such a steady state solution always exists.

An irreducible CTMC is said to be 
\emph{reversible} if there exists a probability mass function $(p(w))_{w \in \script{S}}$ 
that satisfies the following \emph{detailed balance equations}:
\begin{equation} \label{eq:detail-general} 
p(w)q_{w,z} = p(z)q_{z,w} \quad \forall w, z \in \mathcal{S}
\end{equation} 
It is well known that if a probability mass function $(p(w))_{w \in \script{S}}$ solves the detailed balance equations, then it also satisfies the global balance equations and hence is the unique state state.  Indeed, if \eqref{eq:detail-general} holds then for each $w \in \script{S}$ we can sum \eqref{eq:detail-general} over all $z$ to obtain:  
$$  \sum_{z \in \script{S}} p(w)q_{w,z} = \sum_{z \in \script{S}} p(z)q_{z,w}  $$
and thus \eqref{eq:GBE} holds. 
   However, not all 
CTMCs are reversible. That is, not all CTMCs have steady states that satisfy \eqref{eq:detail-general}. 

\subsection{Steady state analysis for non-persistent users}

 It is not obvious whether or not the Markov chain for our system of non-persistent users 
 is reversible.  Fortunately, the system is similar to an \emph{open migration process} with reversibility properties as 
described in \cite{kelly-reversibility}.   An open migration process is a system with $k$ \emph{colonies} that can be described by a Markov chain of the type $(X_1(t), ..., X_k(t))$, where $X_i(t)$ is the current population of colony $i$, transitions between states occur when a single member of colony $i$ moves to colony $j$, and transition rates for such events depend only on the current population $X_i$.  Such a migration process can \emph{almost} be used to model the multi-access system of interest, where the number of type $i$ jobs currently using channels can intuitively be viewed as the population of ``colony $i$.''  However, 
the multi-access system is not a migration system because transition structure is different and 
transition rates depend on the \emph{sum} population $x_1 + ... + x_k$. Nevertheless, reversibility properties of the current system can be established.  To this end, define 
\begin{align*}
\rho_i &= \lambda_i/\mu_j \quad \forall i \in \{1, ..., k\} \\
\rho &= \sum_{i=1}^k \rho_i 
\end{align*}

\begin{figure}[t]
   \centering
   \includegraphics[height=1in]{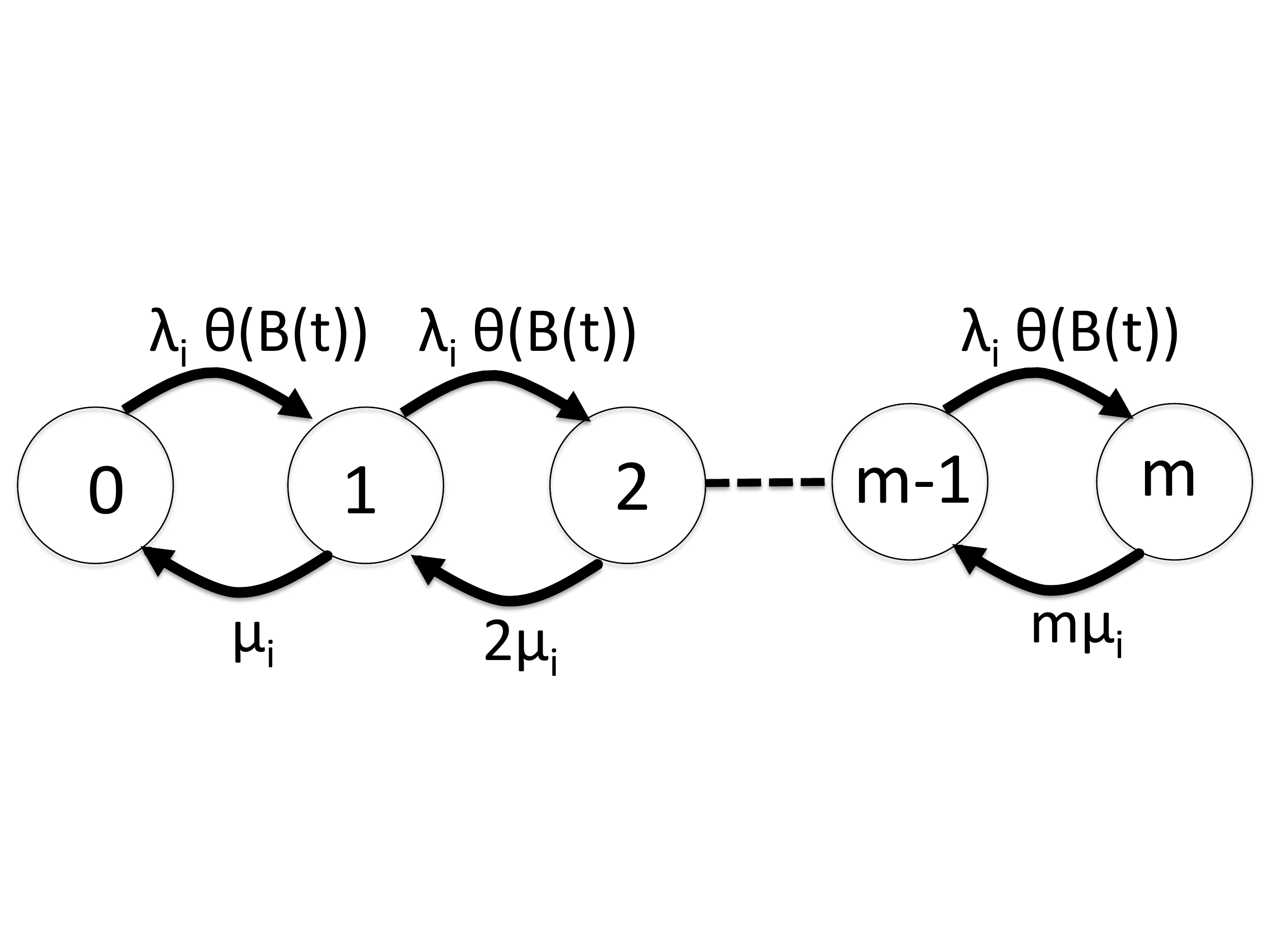} % requires the graphicx package
   \caption{A birth-death-like state diagram of a single user class $i \in \{1, \ldots, k\}$.}
   \label{fig:birth-death-like}
\end{figure}

The technique behind the next theorem is to \emph{guess} a probability mass function and then show the guess satisfies \eqref{eq:detail-general}.  The structure of the guess is not obvious. However, to gain intuition, note that we constructed our guess for the probability of state $(x_1, \ldots, x_k)$
by observing the ``birth-death-like'' structure of the system in Fig. \ref{fig:birth-death-like} and guessing that steady state is a product of terms that include factors of the type $\rho_i^{x_i}/x_i!$ (which are also factors in the steady state mass function of a 
1-dimensional $M/M/\infty$ queue) as well as factors that multiply the chain of  success probabilities 
$\theta(r)$ over all $r \in \{0, \ldots, x_1+...+x_k-1\}$.  Once a good guess is made, it is not difficult to verify the guess satisfies the detailed balance equations.

\begin{thm} \label{thm:1} Under this non-persistent user model with any success probability function $\theta(b)$ that 
satisfies \eqref{eq:sp-1}-\eqref{eq:sp-3}  we have 

a) The CTMC is reversible and the unique steady state distribution is  
\begin{equation} 
p(x_1, ..., x_k) = A \cdot \left(\prod_{r=0}^{-1+x_1+...+x_k}\theta(r)\right)\prod_{i=1}^k \frac{\rho_i^{x_i}}{x_i!} \quad , \forall (x_1, ..., x_k) \in \mathcal{S} \label{eq:distribution1}
\end{equation} 
where 
$A$ is the positive constant that makes the probabilities sum to 1, and we use the convention that $\prod_{r=0}^{-1} \theta(r)=1$. 

b) The steady state probability that there are $b$ channels in use is
$$ P\left[\sum_{i=1}^k X_i=b\right] =  A \cdot \left(\prod_{r=0}^{b-1}\theta(r) \right) \frac{\rho^b}{b!}  \quad \forall b \in \{0, 1, \ldots, m\} $$
where $(X_1, ..., X_k)$ represents a random state vector with distribution equal to the steady state distribution. 

c) The constant $A$  is equal to 
\begin{equation} 
A = p(0,0,...,0)=\frac{1}{\sum_{b=0}^m \left(\prod_{r=0}^{b-1}\theta(r) \right) \frac{\rho^b}{b!}} \label{eq:A} 
\end{equation} 
\end{thm} 

\begin{proof} 
Define the mass function $p(x)$ according to \eqref{eq:distribution1}.  It suffices to show that this $p(x)$ mass function satisfies \eqref{eq:detail-general}.  Since 
there are only two types of possible transitions, it suffices to show that 
$$ p(x) \lambda_j\theta\left(\sum_{i=1}^kx_i\right) = p(x+e_j)(x_j+1)\mu_j  \quad \forall x, x + e_j \in \mathcal{S} $$
It is easy to verify that this equation holds for $p(x)$ as given in the statement of the theorem.  This proves part (a). 

To prove (b), we have 
\begin{align*}
P\left[\sum_{i=1}^k X_i = b\right] &=\sum_{x  \in \mathcal{S} : (x_1+...+x_k)=b} p(x) \\
&=\sum_{x  \in \mathcal{S} : (x_1+...+x_k)=b}A \cdot \left(\prod_{r=0}^{b-1}\theta(r)\right) \prod_{i=1}^k\frac{\rho_i^{x_i}}{x_i!}\\
&= A \cdot \left(\prod_{r=0}^{b-1}\theta(r)\right) \underbrace{\sum_{x \in \mathcal{S} : (x_1+...+x_k)=b} \left(\prod_{i=1}^k \frac{\rho_i^{x_i}}{x_i!}\right)}_{(\rho_1+...+\rho_k)^b/b!}
\end{align*}
where we have used the multinomial expansion:
$$ (\rho_1+ \ldots+\rho_k)^b = \sum_{x \in \mathcal{S}:(x_1+...+x_k)=b} \left(\frac{b!}{x_1!x_2!...x_k!}\right)\prod_{i=1}^k \rho_i^{x_i}$$
This proves part (b).  Part (c) immediately follows from part (b). 
\end{proof} 

The success probability of each newly arriving job depends on the 
current state of the system and not on the class of that job. Since all jobs arrive as Poisson arrivals, and Poisson arrivals see time averages (``PASTA,'' see, for example, \cite{bertsekas-data-nets}),  it follows that jobs of all classes $i \in \{1, ..., k\}$ 
 have the same \emph{long term success probability} for finding an available channel. 
Define $\phi$ as this long term success probability. 
Specifically, if $(x_1, ..., x_k)$ represents a random vector with distribution given by the steady state distribution $p(x)$ in Theorem \ref{thm:1}, then $\phi$ is defined
$$ \phi =  P[\mbox{success}] = \sum_{x \in \mathcal{S}} P[\mbox{success} |(x_1,...,x_k)=x]p(x) $$
With this definition of the success probability $\phi$, 
the long term rate of accepted jobs of type $i$ is $\lambda_i \phi$ jobs/time, 
and the long term rate of dropped jobs of type $i$ is $\lambda_i (1-\phi)$ jobs/time. 
Remarkably, the value of $\phi$ depends only on $\rho$, not on the individual $\rho_i$ values, as shown in the following corollary. 

\begin{cor} \label{cor:non-persistent} For any success probability function $\theta(b)$ that 
satisfies \eqref{eq:sp-1}-\eqref{eq:sp-3}, the long term 
success probability $\phi$ is given by 
\begin{equation} \label{eq:phi-nonpersistent}
 \phi = A  \cdot \sum_{b=0}^m\left(\prod_{r=0}^{b}\theta(r)\right)\frac{\rho^b}{b!} 
 \end{equation} 
where  $A$ is the constant defined in  \eqref{eq:A}.
\end{cor} 
\begin{proof} 
The long term success probability is given by
\begin{align*}
\phi &= P[\mbox{success}]= \sum_{b=0}^{m} \underbrace{P\left[\mbox{success}|\sum_{i=1}^kX_i=b\right]}_{\theta(b)} P\left[\sum_{i=1}^kX_i=b\right] 
\end{align*}
and the result follows by substituting the expression for $P\left[\sum_{i=1}^kX_i=b\right]$ from part (b) of Theorem \ref{thm:1}. 
\end{proof} 

\subsection{Plots for example cases}

\begin{figure}[htbp]
   \centering
   \includegraphics[height=3in]{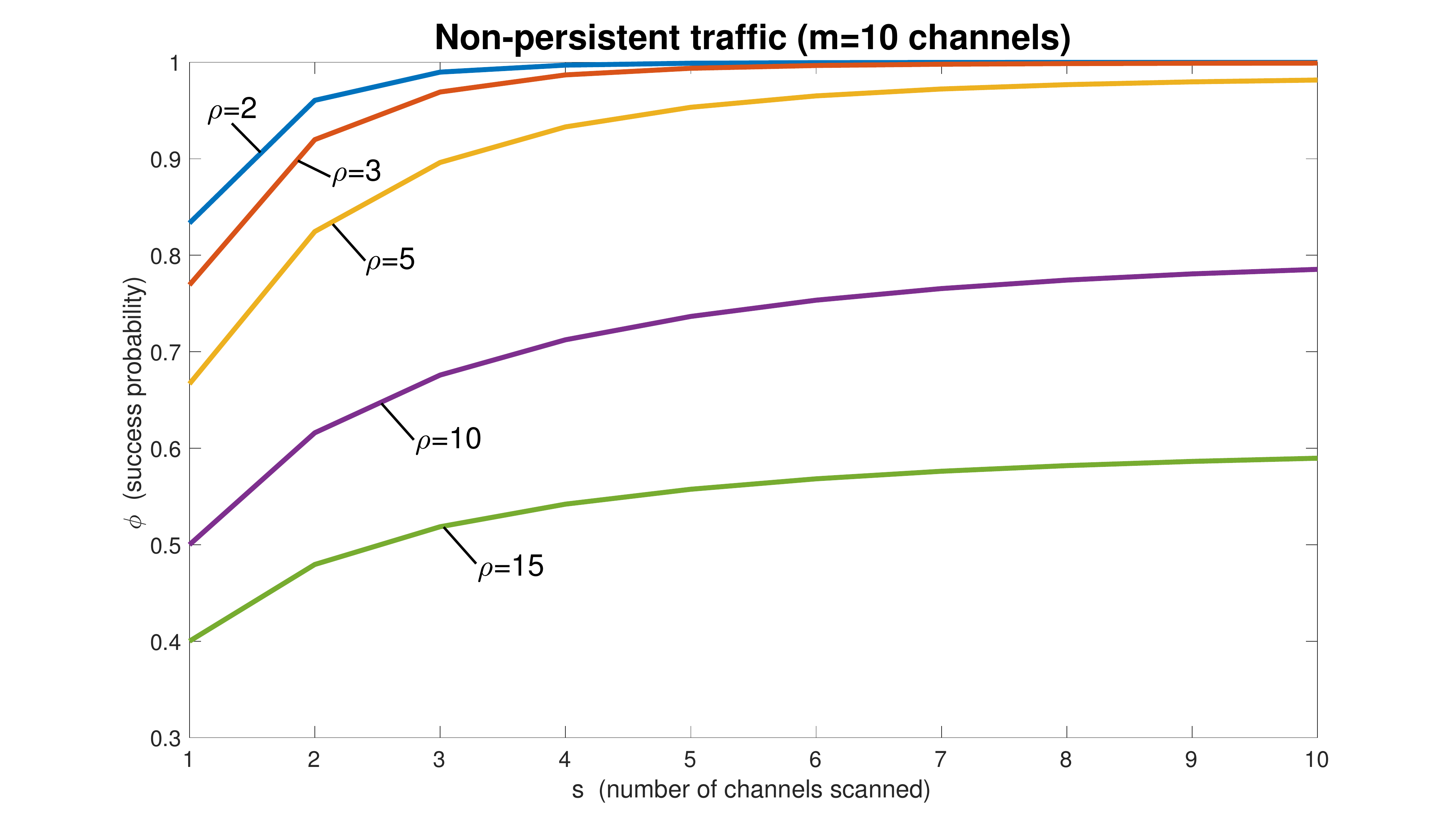} % requires the graphicx package
   \caption{Success probability for non-persistent traffic with $m=10$ channels.}
   \label{fig:plot1-10}
\end{figure}

\begin{figure}[htbp]
   \centering
   \includegraphics[height=3in]{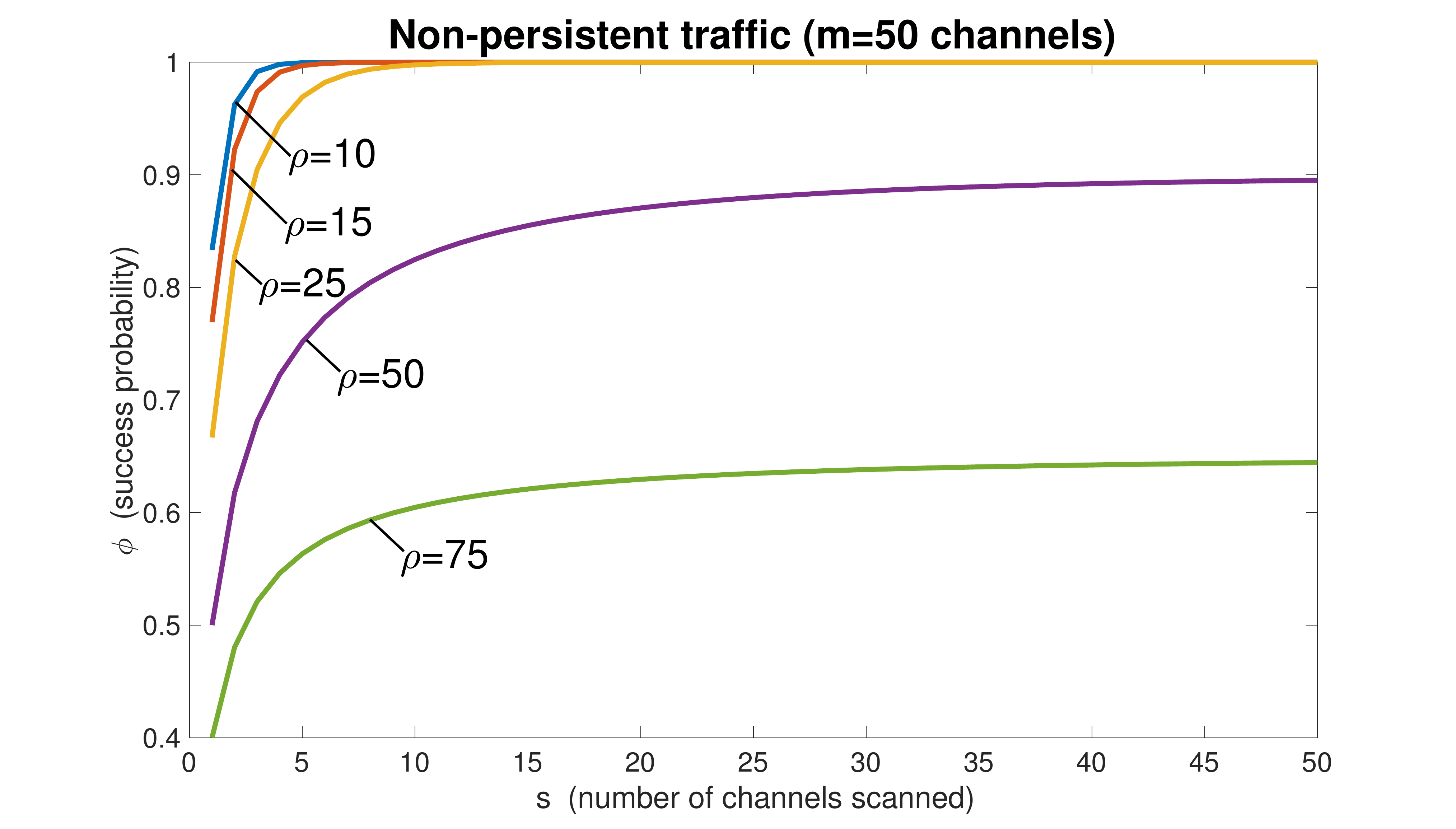} % requires the graphicx package
   \caption{Success probability for non-persistent traffic with $m=50$ channels.}
   \label{fig:plot1-50}
\end{figure}

\begin{figure}[htbp]
   \centering
   \includegraphics[height=3in]{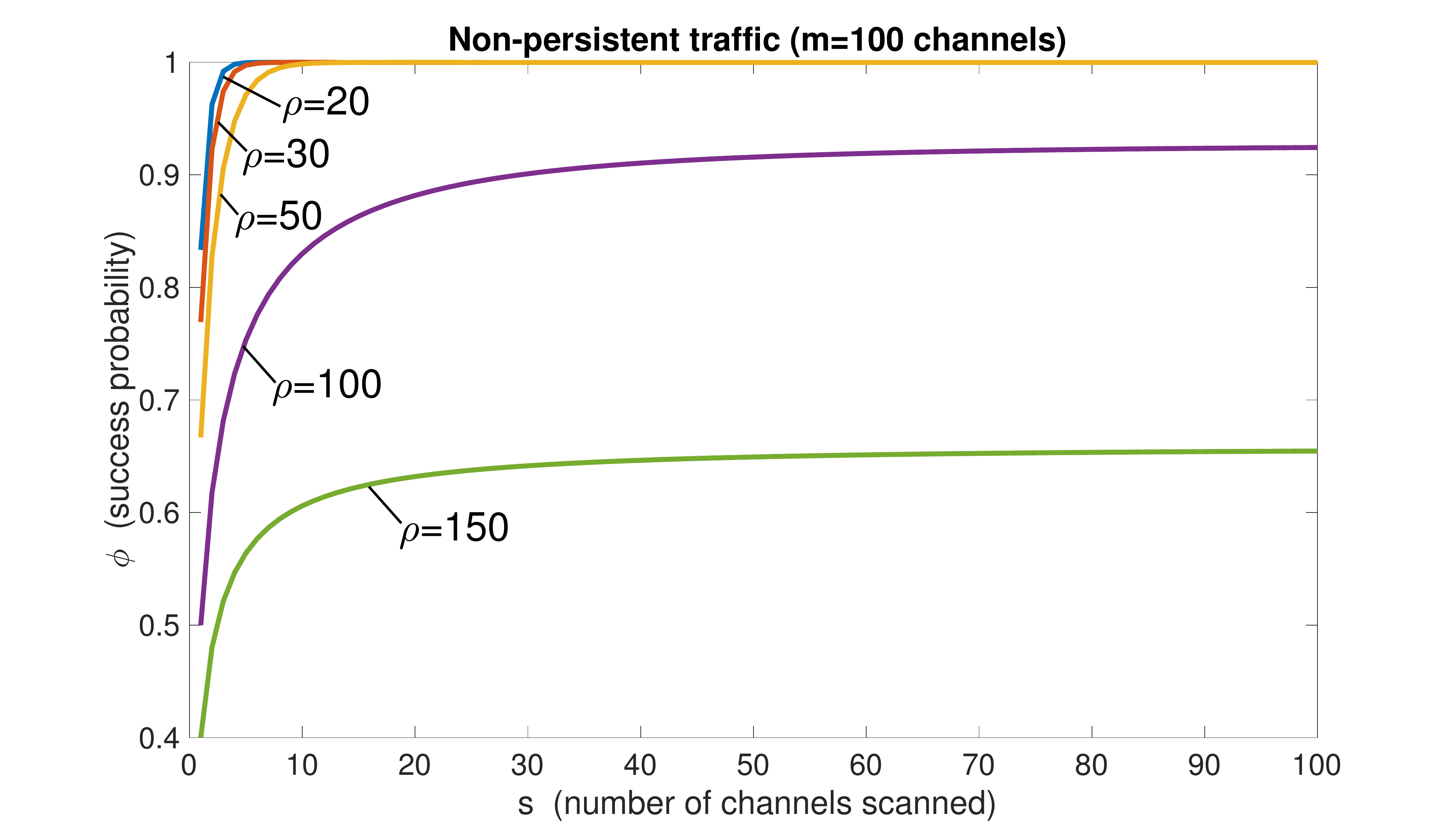} % requires the graphicx package
   \caption{Success probability for non-persistent traffic with $m=100$ channels.}
   \label{fig:plot1-100}
\end{figure}

Corollary \ref{cor:non-persistent} shows that the success probability $\phi$ depends only on the conditional success probabilities $\theta(b)$ and on
the loading parameter $\rho=\sum_{i=1}^k \lambda_i/\mu_i$.  This is an insightful result:   We can understand the success probability through the 
single parameter $\rho$, regardless of the number of classes $k$ of non-persistent users and regardless of the specific $\lambda_i$ and $\mu_i$ parameters for
each class $i \in \{1, \ldots, k\}$.  Notice that, by Little's theorem,  $\rho=\sum_{i=1}^k \lambda_i/\mu_i$ is equal to the steady state 
average number of actively transmitting users  there would be in a virtual system with \emph{infinite resources}: The virtual system has an infinite number of servers, 
each new file of the virtual system receives its own server with probability 1, and no files are dropped.

Fig. \ref{fig:plot1-10} plots the success probability $\phi$ versus $s$ (the number of channels that each user scans) for the case of $m=10$ channels and 
using the $\theta(b)$ probabilities given in \eqref{eq:example-success-probability}.  The values $\rho \in \{2, 3, 5, 10, 15\}$ are shown.  The case $\rho=10$ is when
the average number of active users in a virtual system with infinite resources is equal to 10, the number of channels in the actual system.  This can be 
viewed as a threshold case: When $\rho$ exceeds $m$  (as plotted for the case $\rho=15$ in Fig. \ref{fig:plot1-10}) 
then success probability is necessarily strictly less than $1$ even when the 
number of channels sensed is equal to $m$.   On the other hand, by choosing $s=2$ we obtain a success probability above $0.8$ when $\rho = 2$ or $\rho=3$. 

Better performance is obtained when the number of channels is increased while the $\rho$ values are increased by the same factor:  Figs. \ref{fig:plot1-50} and \ref{fig:plot1-100} show performance for the case $m=50$ channels and $m=100$ channels, respectively, with corresponding $\rho$ values that maintain the same ratio of $\rho/m$ as in the first figure.   
It can be seen that success probability increases to near 1 when the loading is small ($\rho/m \leq 1/2$).  In all of the plots of Figs. \ref{fig:plot1-10}-\ref{fig:plot1-100} it can be seen that success probability is relatively flat for large values of $s$:  A considerable
amount of energy can be saved by just scanning a small subset of the total number of channels.

\section{Persistent and non-persistent users} \label{section:persistent} 

Fix $n$ as a positive integer and suppose that, in addition to the $k$ classes of non-persistent users,  there are $n$ individual persistent users 
with activity states $A_j(t) \in \{I, W, T\}$ and 
behavior parameters  $\alpha_j, \beta_j, u_j, v_j$, as shown in 
Fig. \ref{fig:3-state}.    The $k$ classes of non-persistent users have parameters $\lambda_i$ and $\mu_i$ for all $i \in \{1, \ldots, k\}$. The values of all parameters $\lambda_i, \mu_i$, $\alpha_j, \beta_j,u_j, v_j$ for $i \in \{1, \ldots, k\}$ and $j \in \{1, \ldots, n\}$ are assumed to be positive. 

Recall that $X_i(t)$ is the current number of non-persistent users of type $i$ transmitting, for $i \in \{1, \ldots, k\}$.  The system state is $W(t) = (X_1(t), \ldots, X_k(t); A_1(t), \ldots, A_n(t))$.   The total number of busy channels is 
$$B(t) = \sum_{i=1}^k X_i(t) + \sum_{j=1}^n 1_{\{A_j(t)=T\}}$$ 
where $1_{\{A_j(t)=T\}}$ is an indicator function that is 1 if persistent user $j$ is transmitting at time $t$, and $0$ else.  Let $\theta(b)$ be a success probability  
function defined for $b \in \{0, 1, \ldots, m\}$  that satisfies \eqref{eq:sp-1}-\eqref{eq:sp-3}  (an example $\theta(b)$ function is in \eqref{eq:example-success-probability}). 
As before, if any user attempts access at a  time $t$ such that $B(t)=b$, its conditional success probability is $\theta(b)$.  Notice from Fig. \ref{fig:3-state} 
that the transition rates for the $W \rightarrow T$ transitions of each persistent user depend on the current value of $B(t)$. 

\subsection{Markov chain model} 

Let $\script{S}$ be the state space of the system: This is the set of all vectors $w=(x;a)$, where 
$x=(x_1, \ldots, x_k)$ and $a=(a_1, \ldots, a_n)$, such that $x_i  \in \{0, 1, 2, \ldots\}$ for all $i \in \{1, \ldots, k\}$,  $a_j \in \{I, W, T\}$ for all $j \in \{1, \ldots, n\}$, and 
$$ \sum_{i=1}^k x_i  + \sum_{j=1}^n 1_{\{a_j=T\}}\leq m$$
For simplicity of notation, for each $w \in \script{S}$ define $\busy(w)$ as the number of busy channels associated with state $w$: 
$$ \busy(w) = \sum_{i=1}^k x_i + \sum_{j=1}^n 1_{\{a_j=T\}}$$
To completely describe the transition rates of this CTMC, let $w=(x;a)$ and $z=(y;b)$ be two distinct states in $\script{S}$.  There are three types of transitions that can occur between states $w$ and $z$:  

\begin{itemize} 
\item Non-persistent user $i \in \{1, \ldots, k\}$ ($x_i\leftrightarrow x_i+1$): Recall that $e_i=(0, ..., 0, 1, 0, ..., 0)$ is a vector of size $k$ with a $1$ in component $i$ and zeros in all other components. 

\begin{itemize}  
\item Transitions $(x;a)\rightarrow (x+e_i; a)$ have rate: 
$$ q_{w,z} = \lambda_i\theta(\busy(w))$$

\item Transitions $(x+e_i; a)\rightarrow (x;a)$ have rate 
$$  q_{z,w} = (x_i+1)\mu_i $$ 
\end{itemize} 

\item Persistent user $j \in \{1, \ldots, n\}$ ($I \leftrightarrow W$): 
\begin{itemize} 
\item Transitions $(x, a_1, \ldots, a_{j-1}, I, a_{j+1}, \ldots, a_n) \rightarrow (x, a_1,\ldots, a_{j-1}, W, a_{j+1}, \ldots, a_n)$ have rate: 
$$ q_{w,z} = \alpha_j$$
\item Transitions $(x, a_1,\ldots, a_{j-1}, W, a_{j+1}, \ldots, a_n) \rightarrow (x, a_1, \ldots, a_{j-1}, I, a_{j+1}, \ldots, a_n)$ have rate: 
$$ q_{z,w} = \beta_j $$
\end{itemize} 
\item Persistent user $j \in \{1, \ldots, n\}$ ($W\leftrightarrow T$): 
\begin{itemize} 
\item Transitions $(x, a_1, \ldots, a_{j-1}, W, a_{j+1}, \ldots, a_n) \rightarrow (x, a_1,\ldots, a_{j-1}, T, a_{j+1}, \ldots, a_n)$ have rate: 
$$ q_{w,z} =u_j \theta(\busy(w))$$
\item Transitions $(x_0, a_1,\ldots, a_{j-1}, T, a_{j+1}, \ldots, a_n) \rightarrow (x_0, a_1, \ldots, a_{j-1}, W, a_{j+1}, \ldots, a_n)$ have rate: 
$$ q_{z,w} = v_j $$
\end{itemize} 
\end{itemize} 
It is not difficult to show that the CTMC is irreducible. Indeed, every state can reach the state
$(0, I, I, I, ..., I)$ from a sequence of transitions  that includes no new arrivals, has each transmitting user finish, and has all  persistent users eventually move to the Idle state.  Likewise, the state
$(0,I,I,I, ..., I)$ can reach every state in $\script{S}$. 

\subsection{Steady state probabilities} 

Motivated by the ``birth-death-like'' structure of the persistent user dynamics shown in Fig. \ref{fig:3-state} and by the structure of the steady state probabilities for the non-persistent user case, 
we make the following guess about steady state:  With $w=(x_1, \ldots, x_k; a_1, \ldots, a_n)$ we suggest 
\begin{equation} \label{eq:steady-state-2}
p(w) = B \cdot \left(\prod_{r=0}^{\busy(w)-1}\theta(r)\right)\left(\prod_{i=1}^k\frac{\rho_i^{x_i}}{x_i!}\right)\prod_{j=1}^n \left(\frac{\alpha_j}{\beta_j}\right)^{1_{\{a_j=W\}}} \left(\frac{\alpha_ju_j}{\beta_jv_j}\right)^{1_{\{a_j=T\}}} \quad \forall w \in \script{S}
\end{equation} 
where $\rho_i = \lambda_i/\mu_i$ for all $i \in \{1, \ldots, k\}$ and $B$ is a constant that makes all probabilities sum to 1. 

\begin{thm} \label{thm:system-2} The CTMC for this system with persistent and non-persistent users is reversible and the steady state distribution is given by \eqref{eq:steady-state-2}.
\end{thm} 

\begin{proof} 
It suffices to show that $p(w)$ defined by \eqref{eq:steady-state-2} satisfies the 
detailed balance equations. Consider states $w$ and $z$ in $\script{S}$.  We consider the three possible transition types: 
\begin{itemize} 
\item Non-persistent users ($x_i\leftrightarrow x_i+1$): For simplicity of notation we consider these transitions for non-persistent class $1$ (the result is similar for a general non-persistent class $i \in \{1, \ldots, k\}$). Fix $w =(x_1, \ldots, x_k; a_1, \ldots, a_n)$ and $z =(x_1+1, \ldots, x_k; a_1, \ldots, a_n)$, both being states in $\script{S}$.  Then
\begin{align*}
p(w)q_{w,z} &= \left[B \cdot \left(\prod_{r=0}^{\busy(w)-1}\theta(r)\right)\left(\prod_{i=1}^k\frac{\rho_i^{x_i}}{x_i!}\right)\prod_{j=1}^n \left(\frac{\alpha_j}{\beta_j}\right)^{1_{\{a_j=W\}}} \left(\frac{\alpha_ju_j}{\beta_jv_j}\right)^{1_{\{a_j=T\}}}\right]\lambda_1\theta(\busy(w))\\
&= \left[B \cdot \left(\prod_{r=0}^{\busy(w)}\theta(r)\right)\left(\frac{\rho_1^{x_1+1}}{(x_1+1)!}\right)\left(\prod_{i=2}^k\frac{\rho_i^{x_i}}{x_i!}\right)\prod_{j=1}^n \left(\frac{\alpha_j}{\beta_j}\right)^{1_{\{a_j=W\}}} \left(\frac{\alpha_ju_j}{\beta_jv_j}\right)^{1_{\{a_j=T\}}}\right](x_1+1)\mu_1\\
&= p(z)q_{z,w}
\end{align*}

\item Persistent user $j \in \{1, \ldots, n\}$ ($I\leftrightarrow W$): For simplicity of notation we consider these transitions for persistent user $1$ (the result is the same for a general persistent user $j \in \{1, \ldots, n\}$).   
Fix $w = (x_1, \ldots, x_k;  I, a_2, \ldots, a_n)$ and $z = (x_1, \ldots, x_k; W, a_2, \ldots, a_n)$, both being states in $\script{S}$. Then
\begin{align*}
p(w)q_{w,z} &= \left[B \cdot \left(\prod_{r=0}^{\busy(w)-1}\theta(r)\right)\left(\prod_{i=1}^k\frac{\rho_i^{x_i}}{x_i!}\right)\prod_{j=2}^n \left(\frac{\alpha_j}{\beta_j}\right)^{1_{\{a_j=W\}}} \left(\frac{\alpha_ju_j}{\beta_jv_j}\right)^{1_{\{a_j=T\}}}\right]\alpha_1\\
&= \left[B \cdot \left(\prod_{r=0}^{\busy(w)-1}\theta(r)\right)\left(\prod_{i=1}^k\frac{\rho_i^{x_i}}{x_i!}\right)\prod_{j=2}^n \left(\frac{\alpha_j}{\beta_j}\right)^{1_{\{a_j=W\}}} \left(\frac{\alpha_ju_j}{\beta_jv_j}\right)^{1_{\{a_j=T\}}}\right]\left(\frac{\alpha_1}{\beta_1}\right)\beta_1\\
&= p(z)q_{z,w}
\end{align*}

\item Persistent user $j \in \{1, \ldots, n\}$ ($W\leftrightarrow T$): For simplicity of notation we consider these transitions for persistent user $1$ (the result is the same for a general persistent user $j \in \{1, \ldots, n\}$). Fix $w = (x_1, \ldots, x_k;  I, a_2, \ldots, a_n)$ and $z = (x_1, \ldots, x_k; W, a_2, \ldots, a_n)$, both being states in $\script{S}$. Then
\begin{align*}
p(w)q_{w,z} &= \left[B\cdot \left(\prod_{r=0}^{\busy(w)-1}\theta(r)\right)\left(\prod_{i=1}^k\frac{\rho_i^{x_i}}{x_i!}\right)\left(\frac{\alpha_1}{\beta_1}\right)\prod_{j=2}^n \left(\frac{\alpha_j}{\beta_j}\right)^{1_{\{a_j=W\}}} \left(\frac{\alpha_ju_j}{\beta_jv_j}\right)^{1_{\{a_j=T\}}}\right]u_1\theta(\busy(w))\\
&= \left[B \cdot \left(\prod_{b=0}^{\busy(w)}\theta(r)\right) \left(\prod_{i=1}^k\frac{\rho_i^{x_i}}{x_i!}\right)\left(\frac{\alpha_1 u_1}{\beta_1 v_1}\right)\prod_{j=2}^n \left(\frac{\alpha_j}{\beta_j}\right)^{1_{\{a_j=W\}}} \left(\frac{\alpha_ju_j}{\beta_jv_j}\right)^{1_{\{a_j=T\}}}\right]v_1\\
&= p(z)q_{z,w}
\end{align*}
\end{itemize} 
\end{proof} 

The steady state probabilities in the above theorem can be simplified by aggregating all non-persistent users. 
Consider a state $w = (x_1, \ldots, x_k; a_1, \ldots, a_n) \in \script{S}$. Define $x=\sum_{i=1}^k x_i$ as the number of non-persistent users associated with this state. 
Define $q(x;a_1, \ldots, a_n)$ as the steady state probability that the total number of non-persistent users is $x$ and the state of the persistent users is $(a_1, \dots, a_n)$.  Define $a=(a_1, \ldots, a_n)$ and define 
$$\busy_p(a) = \sum_{j=1}^n 1_{\{a_j=T\}}$$ 
where the $p$ subscript emphasizes that $\busy_p(a)$ counts the number of busy \emph{persistent} users from the vector $a=(a_1, \ldots, a_n)$.   
In particular, the total number of busy channels for a vector $(x,a)$ is $x+\busy_p(a)$.  A vector $(x,a)$ is said to be a \emph{legitimate vector} if 
$x + \busy_p(a) \leq m$.

\begin{cor}  \label{cor:cor-2} For this system with persistent and non-persistent users  we have for all legitimate vectors $(x,a_1, \ldots, a_n)$: 
\begin{equation} \label{eq:expression3} 
q(x; a_1, \ldots, a_n)  =B \cdot \left(\prod_{r=0}^{x+\busy_p(a)-1}\theta(r)\right)\prod_{j=1}^n \left(\frac{\alpha_j}{\beta_j}\right)^{1_{\{a_j=W\}}} \left(\frac{\alpha_ju_j}{\beta_jv_j}\right)^{1_{\{a_j=T\}}}\frac{\rho^x}{x!} 
\end{equation} 
where $B$ is the same constant used in Theorem \ref{thm:system-2}, $\rho = \sum_{i=1}^k \rho_i$, and  $\rho_i = \lambda_i/\mu_i$ for $i\in\{1, \ldots, k\}$.  Further, in the special case when there are no non-persistent users (so that $\rho \rightarrow 0$) we have 
$$ q(a_1, \ldots, a_n) = B \cdot \left(\prod_{r=0}^{\busy_p(a)-1}\theta(r)\right)\prod_{j=1}^n \left(\frac{\alpha_j}{\beta_j}\right)^{1_{\{a_j=W\}}} \left(\frac{\alpha_ju_j}{\beta_jv_j}\right)^{1_{\{a_j=T\}}}  $$
\end{cor} 

\begin{proof} 
From \eqref{eq:steady-state-2} we have: 
\begin{align*}
q(x, a_1, \ldots, a_n) &= \sum_{w\in \script{S}: \sum_{i=1}^kx_i=x} p(w) \\
&=B \cdot \left(\prod_{r=0}^{x+\busy_p(a)-1}\theta(r)\right)\prod_{j=1}^n \left(\frac{\alpha_j}{\beta_j}\right)^{1_{\{a_j=W\}}} \left(\frac{\alpha_ju_j}{\beta_jv_j}\right)^{1_{\{a_j=T\}}}\underbrace{\sum_{x_i: \sum_{i=1}^kx_i=x}\prod_{i=1}^k \frac{\rho_i^{x_i}}{x_i!}}_{\frac{(\rho_1 + \ldots + \rho_k)^x}{x!}}
\end{align*}
Thus, by defining $\rho = \sum_{i=1}^k\rho_i$, we have 
$$ q(x, a_1, \ldots, a_n) =B \cdot \left(\prod_{r=0}^{x+\busy_p(a)-1}\theta(r)\right)\prod_{j=1}^n \left(\frac{\alpha_j}{\beta_j}\right)^{1_{\{a_j=W\}}} \left(\frac{\alpha_ju_j}{\beta_jv_j}\right)^{1_{\{a_j=T\}}}\frac{\rho^x}{x!} $$
\end{proof}

\subsection{Solution complexity} 

The formulas \eqref{eq:steady-state-2} and \eqref{eq:expression3} establish steady state 
probabilities for a very large number of system states. 
The number of states grows exponentially in the problem size.  For example, just considering the 3 possibilities $I, W, T$ for each persistent user, 
we find the number of states is at least $3^{\min[n,m]}$. 
If $\min[n,m] \geq 180$ then $3^{\min[n,m]} \geq 10^{85}$, meaning that the number of states is larger than the current estimate for 
the number of atoms in the universe.  Thus, it is not immediately clear how to compute the constant $B$, and how to 
use the formulas \eqref{eq:steady-state-2} and \eqref{eq:expression3}     to calculate 
things such as the marginal fraction of time that persistent user 1 is busy, the throughput and success probability of persistent user 1, and the throughput and 
success probabilities of the different classes of non-persistent users.  
For some problems that involve reversible networks, such as the admission control problems in 
\cite{bertsekas-data-nets} that are solved by truncation of $M/M/\infty$ queues, 
it can be shown that even calculating the proportionality constant $B$ to within a reasonable approximation is NP-hard \cite{normalizing-coeff-reversible} (see also
\cite{factor-graph-tatikonda-TON} for factor graph approximation methods).   Fortunately, 
our problem has enough structure to allow efficient computation of all of these things via a discrete Fourier transform.  That is developed  next.

\subsection{Calculating $B$} 

Define 
$$ g = \min[m,n]$$
We can sum the probabilities in \eqref{eq:expression3} by grouping states into those that have $b$
persistent users that are busy, for $b \in \{0, 1, \ldots, g\}$, and $x$ non-persistent users: 
\begin{align}
1 &= \sum_{b=0}^g\sum_{x=0}^{m-b} \sum_{a : \busy_p(a)=b} q(x,a_1, \ldots, a_n) \nonumber \\
&= B\sum_{b=0}^g \sum_{x=0}^{m-b} \frac{\rho^x}{x!}\left(\prod_{r=0}^{x+b-1}\theta(r)\right)\sum_{a:\busy_p(a)=b}\prod_{j=1}^n \left(\frac{\alpha_j}{\beta_j}\right)^{1_{\{a_j=W\}}} \left(\frac{\alpha_ju_j}{\beta_jv_j}\right)^{1_{\{a_j=T\}}} \nonumber\\
&=  B\sum_{b=0}^g \sum_{x=0}^{m-b} \frac{\rho^x}{x!}\left(\prod_{r=0}^{x+b-1}\theta(r)\right)c_b 
 \label{eq:plug1} 
\end{align}
where we define $c_b$ by
\begin{equation} \label{eq:cb} 
c_b = \sum_{a: \busy_p(a)=b} \prod_{j=1}^n \left(\frac{\alpha_j}{\beta_j}\right)^{1_{\{a_j=W\}}} \left(\frac{\alpha_ju_j}{\beta_jv_j}\right)^{1_{\{a_j=T\}}}\quad \forall b \in \{0, 1, \ldots, n\}
\end{equation} 
Notice that these $c_b$ values are defined for all $b \in \{0, 1, \ldots, n\}$, even if the number of persistent users $n$ is larger than the number of channels $m$ (so that only $c_0, \ldots, c_m$ are used in \eqref{eq:plug1}). 

It is difficult to obtain the value of $c_b$ by a direct summation in \eqref{eq:cb} because there are so many terms.  However, 
we can define a related polynomial function $f(z)$ defined for all complex numbers $z \in \mathbb{C}$: 
$$ f(z) = \prod_{j=1}^n \left[1 + \left(\frac{\alpha_j}{\beta_j}\right) + z\left(\frac{\alpha_ju_j}{\beta_jv_j}\right)\right]$$
For any given $z \in \mathbb{C}$, the value $f(z)$ can be easily computed as a product of $n$ (complex-valued) terms. 
We observe that 
$$ f(z) = \sum_{b=0}^n c_b z^b$$
This motivates a discrete Fourier transform approach:  Define $i =\sqrt{-1}$ and define 
$$ C_t = f(e^{\frac{-2\pi i t}{n+1}}) = \sum_{b=0}^n c_b e^{\frac{-2\pi i b t}{n+1}} \quad \forall t \in \{0, 1, \ldots, n\}$$
The sequence $\{C_t\}_{t=0}^n$ is the discrete Fourier transform of $\{c_b\}_{b=0}^n$.  The inverse transform gives
$$ c_b = \frac{1}{n+1}\sum_{t=0}^n C_t e^{\frac{2\pi i tb}{n+1}} \quad \forall b \in \{0, 1, \ldots, n\} $$
Of course, these values of $c_b$ only need to be computed for $b \in \{0, 1, \ldots, g\}$ for use in \eqref{eq:plug1}. 
These findings are summarized in the following lemma. 

\begin{lem} (Calculating $B$)  The value $B$ in Theorem \ref{thm:system-2} and Corollary \ref{cor:cor-2} is 
\begin{equation} \label{eq:B}
B = \frac{1}{\sum_{b=0}^g\sum_{x=0}^{m-b}\frac{\rho^x}{x!} \left(\prod_{r=0}^{x+b-1}\theta(r)\right) c_b}
\end{equation} 
where $g=\min[n,m]$ and $c_b$ is defined
\begin{equation} \label{eq:c-b-calc} 
c_b = \frac{1}{n+1} \sum_{t=0}^n C_t e^{\frac{2 \pi i t b}{n+1}} \quad \forall b \in \{0, 1, \ldots, g\}
\end{equation} 
with $i = \sqrt{-1}$ and with $C_t$ given by
$$ C_t = \prod_{j=1}^n\left[1 + \left(\frac{\alpha_j}{\beta_j}\right) + e^{\frac{-2 \pi i  t}{n+1}}\left(\frac{\alpha_ju_j}{\beta_jv_j}\right)\right] \quad \forall t \in \{0, 1, \ldots, n\}$$
Further, in the special case when there are no non-persistent users (so that $\rho \rightarrow 0$) the constant $B$ is replaced by $\tilde{B}$ given by
\begin{equation} \label{eq:B-tilde}
 \tilde{B} = \frac{1}{\sum_{b=0}^g \left(\prod_{r=0}^{b-1}\theta(r)\right) c_b}
 \end{equation} 
\end{lem} 

\begin{proof} 
The proof is contained in the development immediately preceding  the lemma. 
\end{proof}

\subsection{Performance for the individual persistent users} 

Fix $j \in \{1, \ldots, n\}$. Define $P[I_j]$, $P[W_j]$ and $P[T_j]$ as the steady state probability that 
persistent user $j$ is idle, waiting, or transmitting, respectively (see Fig. \ref{fig:3-state}). Define performance variables 
$\gamma_j$ and $\phi_j$ as follows: 
\begin{itemize} 
\item $\gamma_j$ is the \emph{throughput} of persistent user $j$. This is  
the rate at which this user successfully accesses a channel of the multi-access system. Because all files that successfully  access 
a channel are eventually served, $\gamma_j$ is also the rate of file service for persistent user $j$ and so 
\begin{equation} \label{eq:gamma-j-def} 
 \gamma_j = P[T_j]v_j
 \end{equation} 
\item $\phi_j$ is the \emph{success ratio} of persistent user $j$.  This is the rate of access successes divided by the rate 
of access attempts: 
\begin{equation} \label{eq:phi-j-def} 
\phi_j = \frac{\gamma_j}{P[W_j]u_j}
\end{equation} 
\end{itemize} 
The next two lemmas show that: (i) These values can be obtained in terms of $P[I_j]$; (ii)  The probability $P[I_j]$ can be computed via the discrete Fourier transform. 

\begin{lem} For persistent user $j \in \{1, \ldots, n\}$ we have 
\begin{align*}
P[W_j] &= P[I_j](\alpha_j/\beta_j)\\
P[T_j] &= 1-P[I_j](1+(\alpha_j/\beta_j)) \\
\gamma_j &= v_j - v_j P[I_j](1+(\alpha_j/\beta_j))  \\
\phi_j &= \frac{v_j\beta_j - v_jP[I_j](\beta_j+\alpha_j)}{u_j\alpha_j P[I_j]}
\end{align*}
\end{lem}
\begin{proof} 
We use an argument similar to \emph{cut set equations} for CTMCs (in this case we use cuts on the incompletely described CTMC 
of Fig. \ref{fig:3-state}): Consider the 3-state picture of Fig. \ref{fig:3-state} associated with persistent user $j$.  The total number
of transitions  $I\rightarrow W$ is always within 1 of the number of transitions $W\rightarrow I$.  Hence, the long term time average
rate of transitions $I\rightarrow W$ (in units of transitions/time) is the same as the long term rate for transitions $W\rightarrow I$: 
$$ P[I_j] \alpha_j = P[W_j]\beta_j  \implies P[W_j] = P[I_j](\alpha_j/\beta_j)$$
On the other hand we know $P[I_j]+P[W_j]+P[T_j]=1$ and so 
$$ P[T_j] = 1-P[I_j](1+(\alpha_j/\beta_j)) $$
The resulting values of $\gamma_j$ and $\phi_j$ are  obtained by \eqref{eq:gamma-j-def} and \eqref{eq:phi-j-def}. 
\end{proof}  

\begin{lem} For each persistent user $j \in \{1, \ldots, n\}$, the value of $P[I_j]$ is  
\begin{equation} \label{eq:I-j}
P[I_j]=B \sum_{b=0}^{\min[n-1,m]}\sum_{x=0}^{m-b} \left(\prod_{r=0}^{x+b-1}\theta(r)\right) \frac{\rho^x}{x!}c_{j,b}
\end{equation} 
where $B$ is given in \eqref{eq:B}, $\rho = \sum_{i=1}^k(\lambda_i/\mu_i)$, and  $c_{j,b}$ is defined by 
\begin{equation} \label{eq:c-j-b}  
c_{j,b} = \frac{1}{n} \sum_{t=0}^{n-1}\prod_{l \in \{1, \ldots, n\} \setminus j} \left[1 + \left(\frac{\alpha_l}{\beta_l}\right) + e^{\frac{-2\pi i t}{n}}\left(\frac{\alpha_lu_l}{\beta_lv_l}\right)\right]  e^{\frac{2\pi i t b}{n}} \quad \forall b \in \{0, 1, \ldots, \min[n-1,m]\}
\end{equation} 
In the special case when there are no non-persistent users (so that $\rho \rightarrow 0$) we obtain 
$$ P[I_j] = \tilde{B} \sum_{b=0}^{\min[n-1,m]}\left(\prod_{r=0}^{b-1}\theta(r)\right) c_{j,b}$$
where $\tilde{B}$ is given by \eqref{eq:B-tilde}. 
\end{lem} 

\begin{proof} 
Fix $j \in \{1, \ldots, n\}$. Let $\hat{a}_j=(a_1, \ldots, a_{j-1}, a_{j+1}, \ldots, a_n)$. 
By summing probabilities in \eqref{eq:expression3} we obtain: 
\begin{align*}
P[I_j] &= B\sum_{b=0}^{\min[n-1,m]} \sum_{x=0}^{m-b}\left(\prod_{r=0}^{x+b-1}\theta(r)\right)\frac{\rho^x}{x!}\sum_{\hat{a}_j:\busy_p(\hat{a}_j)=b}\prod_{l \in \{1, ..., n\} \setminus j} \left(\frac{\alpha_l}{\beta_l}\right)^{1_{\{a_l=W\}}}\left(\frac{\alpha_lu_l}{\beta_lv_l}\right)^{1_{\{a_l=T\}}}\\
&=B \sum_{b=0}^{\min[n-1,m]}\sum_{x=0}^{m-b} \left(\prod_{r=0}^{x+b-1}\theta(r)\right) \frac{\rho^x}{x!} c_{j,b}
\end{align*}
where $c_{j,b}$ is defined 
$$ c_{j,b} = \sum_{\hat{a}_j:\busy_p(\hat{a}_j)=b}\prod_{l \in \{1, ..., n\} \setminus j} \left(\frac{\alpha_l}{\beta_l}\right)^{1_{\{a_l=W\}}}\left(\frac{\alpha_lu_l}{\beta_lv_l}\right)^{1_{\{a_l=T\}}} \quad \forall b \in \{0, 1, \ldots, n-1\}  $$
This shows that $P[I_j]$ has the form \eqref{eq:I-j}, although it remains to show these $c_{j,b}$ values are equal to \eqref{eq:c-j-b}. 
Define the polynomial $f_j(z)$ by 
$$ f_j(z) = \prod_{l \in \{1, \ldots, n\} \setminus j} \left[1 + \left(\frac{\alpha_l}{\beta_l}\right) + z\left(\frac{\alpha_lu_l}{\beta_lv_l}\right)\right] = \sum_{b=0}^{n-1}c_{j,b}z^b$$
With $i = \sqrt{-1}$ define 
$$C_{j,t} = f_j(e^{\frac{-2\pi i t}{n}}) =\prod_{l \in \{1, \ldots, n\} \setminus j} \left[1 + \left(\frac{\alpha_l}{\beta_l}\right) + e^{\frac{-2\pi i t}{n}}\left(\frac{\alpha_lu_l}{\beta_lv_l}\right)\right]  \quad \forall t \in \{0, 1, \ldots, n-1\}$$
For this persistent user $j \in \{1, \ldots, n\}$, the sequence $\{C_{j,t}\}_{t=0}^{n-1}$ is the discrete Fourier transform of $\{c_{j,b}\}_{b=0}^{n-1}$. 
By the inverse transform we obtain: 
\begin{equation*} 
c_{j,b} = \frac{1}{n} \sum_{t=0}^{n-1}\underbrace{\prod_{l \in \{1, \ldots, n\} \setminus j} \left[1 + \left(\frac{\alpha_l}{\beta_l}\right) + e^{\frac{-2\pi i t}{n}}\left(\frac{\alpha_lu_l}{\beta_lv_l}\right)\right]}_{C_{j,t}}    e^{\frac{2\pi i t b}{n}} \quad \forall b \in \{0, 1, \ldots, \min[n-1,m]\}
\end{equation*} 
\end{proof}

\subsection{Performance for non-persistent users} 

Recall that non-persistent users arrive according to independent Poisson arrival processes. Since Poisson arrivals see time averages (PASTA, see, for example, \cite{bertsekas-data-nets}) it holds that the fraction of time that non-persistent users of class $i \in \{1, \ldots, k\}$ see a system with $y$ busy channels is the same for all classes $i$ and is equal to the long term 
fraction of time that there are $y$ busy channels.  Hence, all non-persistent users see the same access success probability, call it $\phi_0$: 
$$ \phi_0 = \sum_{y=0}^m P[\mbox{Exactly $y$ busy channels}]\theta(y)$$
where the probability of having exactly $y$ busy channels in the right-hand-side above is the steady state value and can be found by summing the $q()$ probabilities in \eqref{eq:expression3} over all states that yield $y$ busy channels: 
\begin{align*}
P[\mbox{Exactly $y$ busy channels}] &= \sum_{b=0}^{\min[y,n]}\sum_{a:\busy_p(a)=b} q(y-b;a_1, \ldots, a_n)\\
&=B\sum_{b=0}^{\min[y,n]} \frac{{\rho}^{y-b}}{(y-b)!} \left(\prod_{r=0}^{y-1}\theta(r)\right)c_b
\end{align*}
where $c_b$ is defined in \eqref{eq:cb} and is computed via \eqref{eq:c-b-calc};  $B$ is defined in \eqref{eq:B}.  Hence
$$ \phi_0 = B \sum_{y=0}^m\sum_{b=0}^{\min[y,n]} \frac{{\rho}^{y-b}}{(y-b)!} \left(\prod_{r=0}^{y}\theta(r)\right)c_b $$

\section{Validation on test cases}

This section validates the results of the previous section (obtained by the discrete Fourier transform) by considering simple example cases and comparing to simulated performance. 

\subsection{Three identical persistent users} 

Consider the following test case with $m=5$ channels and where each user scans a subset $s=2$ of these channels.  
There is one class of non-persistent traffic with $\lambda_1 = 1$ and $\mu_1=2$, so that $\rho = 1/2$. There are three non-persistent users with identical parameters $\alpha_j=\alpha$, $\beta_j=\beta$, $u_j=u$, $v_j=v$ for all $j \in \{1, 2, 3\}$ with 
$$ \alpha = 1; \beta = 1; u = 5; v=10$$
We compare the exact 
success probabilities $\phi_0, \phi_1, \phi_2, \phi_3$ and the persistent 
user state probabilities $P[I_j]$, $P[W_j]$, $P[T_j]$ obtained by the formulas in the previous section with simulation values. The exact values were calculated using complex number multiplication in matlab (the imaginary parts of all real-valued quantities were indeed found to be zero in the matlab computation).  The simulation was conducted in matlab over a period of $10^7$ transitions of the CTMC.  Access attempts that fail were also counted as transitions, even though these transitions did not change the state of the CTMC.  The success probability ratios were obtained by taking the ratio of accepted attempts to total attempts for each type of user.  To obtain a fast simulation, only discrete events were simulated. The exponentially distributed times between transitions was not simulated. Instead,  the time spent in each state was updated immediately when the simulation entered that state by adding the theoretical average time spent in that state
(equal to the reciprocal of the sum of all possible transition rates out of that state).  The resulting simulated fraction of time being in a particular state is then the sum of all average times computed while in this state divided by the sum of all average times computed.    This simulation technique is similar in spirit to the standard CTMC simulation technique of \emph{uniformization} (see, for example, \cite{ross-prob}). However, in this context, it is easier to implement
because it allows each state to have a different sum of outgoing transition rates, and hence does not require adding  
faux self-transitions to impose uniformization. 

Tables \ref{tab:table1}-\ref{tab:table2} report the results.  The data from these tables shows a good agreement between theory and simulation. 

% Requires the booktabs if the memoir class is not being used
\begin{table}[htbp]
   \centering
   %\topcaption{Table captions are better up top} % requires the topcapt package
   \begin{tabular}{|c|c|c|c|} % Column formatting, @{} suppresses leading/trailing space
   \hline
    Type &  Success prob & Simulated    & Exact    \\ \hline
      Non-persistent & $\phi_0$&   0.9530   & 0.9527 \\
      Persistent & $\phi_1$& 0.9676 & 0.9674    \\
   Persistent &   $\phi_2$ &  0.9675 & 0.9674   \\
  Persistent & $\phi_3$ &  0.9670    & 0.9674 \\ \hline
   \end{tabular}
   \caption{A comparison of simulated and exact values for success probability.}
   \label{tab:table1}
\end{table}

\begin{table}[htbp]
   \centering
   %\topcaption{Table captions are better up top} % requires the topcapt package
   \begin{tabular}{|c|c|c|c|} % Column formatting, @{} suppresses leading/trailing space
   \hline
  Persistent Users & $P[I_j]$ & $P[W_j]$ & $P[T_j]$ \\ \hline
     User 1 Simulation &    0.4028  &    0.4024   & 0.1948 \\ \hline
     User 2 Simulation &  0.4021  &   0.4030  &   0.1949\\ \hline
     User 3 Simulation & 0.4035  &  0.4019    & 0.1946\\ \hline\hline
     Exact   &  0.4026 & 0.4026 & 0.1947 \\ \hline
   \end{tabular}
   \caption{A comparison of simulated and exact values for state probabilities.}
   \label{tab:table2}
\end{table}

\subsection{Two classes of persistent users (three in each class)} 

This test case considers $m=10$ channels and $s=2$. 
This test case considers two classes of persistent users (as defined by the $\alpha_j,\beta_j, u_j,v_j$ parameters)
with three users in each class: 
\begin{itemize} 
\item Persistent class A:  $\alpha = 1; \beta = 1; u = 5; v=10$
\item Persistent class B: $\alpha = 1; \beta=1; u= 5; v=1$
\end{itemize} 
In particular, class A persistent users have the same parameters as the previous subsection, while class B persistent
users have file that take 10 times longer to serve.    There is a single non-persistent class with parameters $\lambda_1=1, \mu_1=1$ (so that $\rho =1$).  

As before, we compare the exact theoretical values from the previous section with simulated values obtained over a simulation 
with $10^7$ transitions. 
The results are shown in Tables \ref{tab:table3}-\ref{tab:table4}. Again there is good agreement between
theory and simulation.  
Notice that class B persistent users spend much more time in the transmitting state.  
These users also have a slightly higher access success probability. Intuitively, this is because the average time
that a  class B users spends transmitting a file over a channel is 10 times longer than class A users and non-persistent users. 
Thus,  when a class B user wants to access a channel, 
this particular class B user is certainly \emph{not currently occupying a channel}, which is one less channel-hogger
to worry about. 

% Requires the booktabs if the memoir class is not being used
\begin{table}[htbp]
   \centering
   %\topcaption{Table captions are better up top} % requires the topcapt package
   \begin{tabular}{|c|c|c|c|} % Column formatting, @{} suppresses leading/trailing space
   \hline
    Type &  Success prob & Simulated    & Exact    \\ \hline
      Non-persistent & $\phi_0$&    0.8819  & 0.8822 \\ \hline
      Persistent A & $\phi_1$& 0.8940 &   0.8937  \\
   Persistent A&   $\phi_2$ & 0.8938 & 0.8937   \\
  Persistent A& $\phi_3$ &    0.8939  & 0.8937 \\ \hline
    Persistent B& $\phi_4$ &  0.9213    & 0.9209 \\ 
  Persistent B& $\phi_5$ &    0.9206  & 0.9209 \\ 
  Persistent B& $\phi_6$ &    0.9207  & 0.9209 \\ \hline
   \end{tabular}
   \caption{A comparison of simulated and exact values for success probability.}
   \label{tab:table3}
\end{table}

\begin{table}[htbp]
   \centering
   %\topcaption{Table captions are better up top} % requires the topcapt package
   \begin{tabular}{|c|c|c|c| |c|c|c|c|} % Column formatting, @{} suppresses leading/trailing space
   \hline
  Persistent Class A & $P[I_j]$ & $P[W_j]$ & $P[T_j]$ & Persistent Class B & $P[I_j]$ & $P[W_j]$ &$P[T_j]$\\ \hline
     User 1 Simulation &  0.4086   & 0.4090  &  0.1823      &User 4 Simulation & 0.1507 &   0.1509 &    0.6984\\ \hline
     User 2 Simulation &   0.4094   &  0.4084  &   0.1822     & User 5 Simulation & 0.1533 &   0.1512    &0.6955 \\ \hline
     User 3 Simulation &  0.4093  &   0.4082 &     0.1825&    User 6 Simulation & 0.1515 &    0.1510  &  0.6976  \\ \hline\hline
     Exact   &  0.4087& 0.4087 &  0.1826 & Exact &0.1514 &0.1514 & 0.6972\\ \hline
   \end{tabular}
   \caption{A comparison of state probabilities for class A and B persistent users.}
   \label{tab:table4}
\end{table}

\section{Conclusion} 

This paper considers a multi-channel multi-access system with $m$ identical channels.  When a user wants to send a new file, it scans only a subset of $s$ channels to see if one is idle.  Scanning fewer than $m$ channels can significantly reduce complexity and energy expenditure.  When all users are non-persistent and send at most one file, 
a simple expression for success probability was derived that depends only on $m$, $s$, and the system loading $\rho=\sum_{i=1}^k \lambda_i/\mu_i$.  This allows for an arbitrarily large number of non-persistent user classes with arbitrary $\lambda_i$ and $\mu_i$ values for each class, provided that the sum of the $\lambda_i/\mu_i$ values is $\rho$. 
It was shown that when $\rho<m/2$ then success probability is very large (near 1) and gets larger when the number of channels $m$ is increased while maintaining the same ratio  
$\rho/m$. 

The case when both persistent and non-persistent users was also analyzed. Each persistent user has its own \emph{activity parameters} and 
behaves according to a 3-state process with idle, waiting, and transmitting states.   The exact steady state values were also derived in this setting. 
This is a more complex scenario and it is not obvious how to obtain individual performance parameters from 
the joint steady state mass function (which an have more terms than the number of atoms in the universe).   An efficient method for obtaining
the individual performance values was developed using a discrete Fourier transform.

 % ------------------------------------------------------------------------
%GATHER{Xbib.bib}   % For Gather Purpose Only
%GATHER{Paper.bbl}  % For Gather Purpose Only
\bibliographystyle{unsrt}
\bibliography{../../../../latex-mit/bibliography/refs}
\end{document}